\documentclass[11pt,reqno]{amsart}
\usepackage{amsmath}
\usepackage{array}

\usepackage{mathabx}
\usepackage{color}

\newtheorem{theorem}{Theorem}
\newtheorem{proposition}[theorem]{Proposition}
\newtheorem{lemma}[theorem]{Lemma}

\theoremstyle{definition}
\newtheorem{definition}[theorem]{Definition}

\theoremstyle{remark}
\newtheorem{remark}[theorem]{Remark}
\input{amssym.def}
\input{amssym}

\newcommand{\C}{\mathbb{C}}
\newcommand{\la}{\langle}
\newcommand{\ra}{\rangle}
\newcommand{\R}{\mathbb{R}}

\newcommand{\vp}{\varphi}
\newcommand{\be}{\begin{equation}}
\newcommand{\ee}{\end{equation}}

\newcommand{\su}{\mathrm{SU}_2}
\newcommand{\ba}[1]{\overline{#1}}

\renewcommand{\phi}{\varphi}
\newcommand{\ci}{C^{\infty}}

\def\bar{\overline}

\makeatletter

\@addtoreset{equation}{section}
\makeatother

\newcommand{\cons}{{C}}
\newcommand{\conss}{{D}}
\newcommand{\psiN}{\psi^N}
\newcommand{\unun}{M}
\newcommand{\cof}{{C}}
\newcommand{\ns}{{s}}
\newcommand{\at}{{$a$-T\"oplitz}\ }

\begin{document}
\baselineskip=18pt
\title{Symbolic calculus for singular curve operators}
\author{Thierry Paul}
\address{CNRS and 
%Centre de Math\'ematiques Laurent Schwartz, 
%\'Ecole Polytechnique, 91128 Palaiseau Cedex 
Laboratoire Jacques-Louis Lions, Sorbonne Universit\'e, 4 place Jussieu 75005 Paris
Fran\-ce}
\email{
%paul@math.polytechnique.fr
thierry.paul@upmc.fr}

\date{}
%\date{\today}
%\maketitle

\begin{abstract}
We define a generalization of the T\"oplitz quantization, suitable for operators whose T\"oplitz symbols are singular. We then show that singular curve operators in Topological Quantum Fields Theory (TQFT) are precisely  generalized T\"oplitz operators of this kind and we compute for some of them, and conjecture for the others, their main symbol, determined by the associated  classical trace function.
\\
\ \\ \ 

\hfill\textit{in  memory of {\bf Erik Balslev}}

\hfill\textit{from whom I learned so much}

\hfill\textit{ in mathematics and in physics}
\end{abstract}
%\date{} 
\maketitle 

%\begin{color}{red}
%A ne pas oublier :
%
%le cas regulier : fait, tout va bien.
%
%le casen decomposant les cosinus : ca marche
%
% lien avec Toplitz ordinaires : a faire
% 
% le symbole explose mias donne lieu a un $\psi^a_z$ integrale en z.
% \vskip 1cm
%\end{color} 
 \tableofcontents

%\newpage

\section{Introduction}\label{intro}
%Plan intro:
% 
% - TQFT, opertqurs courbe, rappel de toe, extension au cas singulier-objet du papier.
% 
% - cela se ramene a matrices etc.
% 
% - de quoi ces matrices sont le Toplitz.
% 
% - description rapide et resultat TQFT
% 
% \newpage
In 1925, Heisenberg invented quantum mechanics as a change of paradigm from (classical) functions  to (quantum) matrices. He founded the new mechanics on the well known identity
$$
\tfrac 1{i\hbar}[Q,P]=1
$$
that, a few months later, Dirac recognized as the quantization of the Poisson bracket
$$
\{q,p\}=1.
$$
Again a few years later, Weyl stated the first general quantization formula by associating to any function $f(q,p)$ the operator
$$
F(Q,P)=\int \tilde f(\xi,x)e^{ i\frac{ xP-\xi Q}\hbar}d\xi dx
$$ 
where $\tilde f$ is the symplectic Fourier transform defined analogously by 
$$
f(q,p)=\int \tilde f(\xi,x))e^{ i\frac{ xp-\xi q}\hbar}d\xi dx.
$$

Many years after was born the pseudodifferential calculus first establish by Calderon and Zygmund, and then formalized by H\"ormander through the formula giving the integral kernel $\rho_F$ for the quantization $F$ of a symbol $f$ in $d$ dimensions as
%$$
%\rho_F(x,y)=
%%(2\pi\hbar)^{-d}
%\int a(x,\xi)e^{i\frac{\xi(x-y)}\hbar}
%%d\xi.
%\frac{d\xi}{(2\pi\hbar)^{d}}.
%%d\xi/(2\pi\hbar)^{d}
%$$
$$
\rho_F(x,y)=\int f(q,p)e^{i\frac{p(x-y)}\hbar}\frac{dp}{(2\pi\hbar)^d}
$$
A bit earlier had appeared, both in quantum field theory  and in optics (Wick quantization) the (positive preserving) T\"oplitz quantization of a symbol $f$
$$
\mbox{Op}^T[f]=\int f(q,p)|q,p\rangle\langle q,p|dqdq
$$ 
where $|q,p\rangle$ are the famous (suitably normalized) coherent states.

As we see, quantization is not unique. But all the different symbolic calculi presented above share, after inversion of the quantization formul\ae\ written above, the same two first asymptotic features:
\begin{itemize}
\item the symbol of a product is, modulo $\hbar$, the product of the symbols
\item the symbol of the commutator divided by $i\hbar$ is, modulo $\hbar$ again, the Poison bracket of the symbols.
\end{itemize}
In other words, they all define a classical underlying space (an algebra of functions) endowed with a Poisson (of more generally symplectic) structure.

But it is very easy to show that this nice quantum/classical picture has its limits. And one can easily construct quantum operators whose classical limit will not follow the two items exprressed above.

Consider for example the well known creation and annihilation operators $a^+=Q+iP,a^-=Q-iP$. They act of the eigenvectors $h_j$ of the harmonic oscillator by
$$
 a^+h_j=\sqrt{(j+\tfrac12)\hbar}h_{j+1},\   a^-h_j=\sqrt{(j-\tfrac12)\hbar}h_{j-1}.
 $$
 Consider now the matrices
 
\[
M_1^+=\begin{pmatrix}
0&0&0&0&\dots&\dots\\
1&0&0&0&\dots&\dots\\
&&\dots&&&\\
&&\dots&&&\\
&&\dots&&&\\
0&\dots&0&1&0&\dots\\
\dots&\dots&\dots&\dots&\dots&\dots
\end{pmatrix}
\] 
and its adjoint
\[
M_1^-=
\begin{pmatrix}
0&1&0&0&\dots&\dots\\
0&0&1&0&\dots&\dots\\
&&\dots&&&\\
&&\dots&&&\\
&&\dots&&&\\
0&\dots&0&0&0&\dots\\
\dots&\dots&\dots&\dots&\dots&\dots
\end{pmatrix},
\]
%and the operator $\mathcal M_1^\pm$ they define on the basis $\{\varphi_n^N,n=o,\dots,N-1\}$ of $\mathcal H_N$ defined by \eqref{basis}.

An elementary computation shows that 
$$
\mathcal M_1^+=a^+(P^2+Q^2)^{-1/2},\ \mathcal M_1^-=(P^2+Q^2)^{-1/2}a^-. 
$$
therefore, their (naively)expected leading symbols are $f^+(q,p)=\sqrt{\frac{q+ip}{q-ip}}$ and $f^-(q,p)=\sqrt{\frac{q-ip}{q+ip}}$ or, in polar coordinates $q+ip=\rho e^{i\theta} $,
%$f^+=\sqrt\frac{e^{i\theta}}{e^{-i\theta}}, f^-=\sqrt\frac{e^{-i\theta}}{e^{i\theta}}$.
$f^\pm=e^{\pm i\theta}$.

 %are T\"oplitz operators, in the sense of Section \ref{canonical}, of exact symbol the functions $f_1^\pm(z)=\sqrt{\frac z{\bar z}}=e^{i\theta}$.

If symbolic calculus would work 
%for these T\"oplitz operators, 
the leading symbol of $\mathcal M_1^+
\mathcal M_1^-$ should be equal to $1$ 
%and a T\"oplitz operator of leading symbol $f_1^+(z)\times f_1^-(z)=1$, 
and $\mathcal M_1^+
\mathcal M_1^-$ should be therefore  close to the identity $I$ as $\hbar\to 0$.

 But
\[
M_1^+M_1^-=
\begin{pmatrix}
0&0&0&0&\dots&\dots\\
0&1&0&0&\dots&\dots\\
&&\dots&&&\dots\\
&&\dots&&&\dots\\
&&\dots&&&\dots\\
0&\dots&0&0&1&\dots\\
\dots&\dots&\dots&\dots&\dots&\dots
\end{pmatrix}=I-|h_0\rangle\langle h_0|\nsim I,
\]
The reason for this defect comes from the fact that the function $e^{i\theta}=\sqrt{\frac z{\bar z}}$ is not a smooth function on the plane. In fact it is not even continuous at the origin: $F(z)$ can tend to any value in $\{e^{i\theta},\theta\in\R\}$ when $z$ tends to zero.

Note finally that the commutator 
\be\label{comm}
[\mathcal M_1^+,\mathcal M_1^-]
=
%\scriptsize
\begin{pmatrix}
-1&0&0&0&\dots&\dots\\
0&0&0&0&\dots&\dots\\
&&\dots&&&\dots\\
&&\dots&&&\dots\\
&&\dots&&&\dots\\
0&\dots&0&0&0&\dots\\
\dots&\dots&\dots&\dots&\dots&\dots
\end{pmatrix}=-|h_0\rangle\langle h_0|\neq O(\hbar),
\small
\ee
so that its symbol doesn't vanish at leading order, as expected by standard symbolic asymptotism.

%Symbolic calculi are a way of establishing an underlying classical obejct to quantum mechanics.
% 
% 
% - historique quantification heisenberg etc
% 
% - symbole produit $\sim$ produit des symboles et commutateur d'ordre $\hbar$
% 
% - symbole de matrices données a lávacmce
% 
% - sur L2 base des Hermites matrices identite, puis oscilla harmonique, puis ceraton annihilation
% 
% - matrice singuliere $h_j\to h_{j\pm 1}$
% 
% - symbole $e{i\theta}$ mais pas pour le produit etc.... commutateur etc
 
 One of the main goal of this paper is to define a quantization procedure which assigns a symbolic calculus to matrices presenting the pathologies analogues to the ones of $\mathcal M_1^+,\mathcal M_1^-$. We will state the results in the framework of quantum mechanics on the sphere $S^2$ as phase space. 
The reason of this is the fact that it is this quantum setting which correspond to the  asymptotism in Topological Quantum Fields Theory (TQFT) studied,  among others, in \cite{mp}.
 
 %We will give a definition of symbol valid for all example quoted before, as 
 
 This procedure will be a non trivial extension to the T\"oplitz (anti-Wick) quantization already mentioned, and we will derive a suitable notion of symbol. 
 
 \vskip 3cm

Indeed, another main goal of this article is to give a semiclassical settings to all curve-operators in TQFT in the case of the once punctured torus or the $4$-times punctured sphere. In \cite{mp} was established that these curve-operators happen, for almost all 
colors associated to the marked points, to be  T\"oplitz operators associated to the quantization of the two-sphere, in some asymptotics of large number of colors. It happens that this result applies for every curve whose classical trace function is a smooth function on the sphere. Since this trace function is shown to be the principal T\"oplitz-symbol of the curve operator, the lack of smoothness ruins the possibility of semiclassical properties for the curve operator in the paradigm of T\"oplitz quantization (see Section \ref{tqft} below for a very short presentation of TQFT and the main results of \cite{mp}).
In fact these singular trace functions are not even continuous at the two poles of the sphere, which suggests a kind of blow-up on the two singularities of the classical phase-space. This is not surprising that such a regularization should be done in a more simple way at the quantum level.

In the present paper we will show how to ``enlarge" the formalism of T\"oplitz operators to ``$a$-T\"oplitz operators", in order to catch the asymptotics of the singular cases by semiclassical methods and compute the leading order symbols of some of them and conjecture them for the general  singular curve operators. This principal symbol will be completely determined by the corresponding classical trace function, but will not be equal to it, for the reason that this enlarged $a$-T\"oplitz quantization procedure involves operator valued symbols. This construction will also be valid in the regular cases, where in this case the operator valued symbol is just a potential, hence it is defined by a function on the sphere whose leading behaviour is given by the trace function, as expected.

Therefore we are able in this paradigm to handle the large coloring asymptotism of \textbf{all} curve operators n the case of the once punctured torus or the $4$-times punctured sphere (note that the method we use is able to  give some partial results in higher genus cases). 

We also study the natural underlying phase-space of our enlarged paradigm, the corresponding moduli space for TQFT, as a non-commutative space by identification with the non-commutative algebra of operator valued functions appearing at the classical limit for the symbol of the $a$-T\"oplitz operators, in the spirit of noncommutative geometry.

We will built the construction of the $a$-T\"oplitz quantization by showing its necessity on some toy matrices situations in Sections \ref{sq} after having defined in Section \ref{hilbert} the new Hilbert space on which these matrices will act, and before to show in Section \ref{tqft} how general curve operators in TQFT enter this formalism.  

{{The main results are Theorems \ref{sigmann} and \ref{onemain}, out of Definition \ref{defsym}   and Theorem \ref{main} together with Section \ref{exemples} below. The $a$-T\"oplitz operators are introduced in Definition \ref{atop}.
%, using formul\ae \eqref{convol},\eqref{cnumber},\eqref{def} and \eqref{dmun}.}}
\vskip 3cm
%The rest of this introduction will try to motivate this change of quantization on some very simple example of matrices which are not standard T\"oplitz operators.

Quantization of the sphere is briefly reviewed in Section \ref{canonical}, we won't repeat it here. Let us just say  that it consists in considering the sphere $S^2$ as the compactification of the plane $\mathbb C$. Hence one expect that the singular phenomenon which appeared above at the origin should now appear twice at the poles of the sphere. The quantum Hilbert space can be represented as the space of entire functions, square-integrable with respect to a measure $d\mu_N$ given in \eqref{munu}. 

%Let us consider the matrix
%\[
%M_1^+=\begin{pmatrix}
%0&0&0&0&\dots&0\\
%1&0&0&0&\dots&0\\
%&&\dots&&&\\
%&&\dots&&&\\
%&&\dots&&&\\
%0&\dots&0&1&0&0\\
%0&\dots&0&0&1&0
%\end{pmatrix}
%\] 
%and its adjoint
%\[
%M_1^-=
%\begin{pmatrix}
%0&1&0&0&\dots&0\\
%0&0&1&0&\dots&0\\
%&&\dots&&&\\
%&&\dots&&&\\
%&&\dots&&&\\
%0&\dots&0&0&0&1\\
%0&\dots&0&0&0&0
%\end{pmatrix},
%\]
%and the operator $\mathcal M_1^\pm$ they define on the basis $\{\varphi_n^N,n=o,\dots,N-1\}$ of $\mathcal H_N$ defined by \eqref{basis}.
%
%An elementary computation shows that $\mathcal M_1^\pm$ are T\"oplitz operators, in the sense of Section \ref{canonical}, of exact symbol the functions $f_1^\pm(z)=\sqrt{\frac z{\bar z}}=e^{i\theta}$.
%
%If symbolic calculus would work for these T\"oplitz operators, $\mathcal M_1^+
%\mathcal M_1^-$ should be a T\"oplitz operator of leading symbol $f_1^+(z)\times f_1^-(z)=1$, and should therefore be close to the identity $I$.
%
% But
%\[
%M_1^+M_1^-=
%\begin{pmatrix}
%0&0&0&0&\dots&0\\
%0&1&0&0&\dots&0\\
%&&\dots&&&\\
%&&\dots&&&\\
%&&\dots&&&\\
%0&\dots&0&0&1&0\\
%0&\dots&0&0&0&1
%\end{pmatrix}=I-|\varphi^N_0\rangle\langle\varphi^N_0|\nsim I,
%\]
%the reason for this defect comes from the fact that the function $e^{i\theta}=\sqrt{\frac z{\bar z}}$ is not a smooth function on the sphere. In fact it is not even continuous at the two poles: $F(z)$ can tend to any value in $\{e^{i\theta},\theta\in\R\}$ when $z$ tends to any of the two poles.

Instead of trying to blow-up these two singularities at a ``classical" (namely manifold) level, we will see that there is an easiest way of solving the problem by working directly at the ``quantum" level. Namely, instead of considering the quantization process related to the so-called coherent state family $\rho_z$ defined in \eqref{repro} and which are (micro)localized at the points $z\in S^2$, we will consider families of states $\psi_z^a:=\int_\R a(t)e^{i\frac{\tau(z)t}\hbar}\rho_{e^{it}z}\frac{dt}{\sqrt{2\pi}}
$
where $\tau(z)=\frac{|z|^2}{1+|z|^2}$ and $a\in\mathcal S(\R)$ (see Section\ref{h1} for details). For $z$ not at the poles, $\psi_z^a$ is a Lagrangian (semiclassical) distribution (WKB state) localized on the parallel passing through $z$ \cite{pu}, but for $z$ close to the pole the states $\psi_z^a$ catches a different information.  The equality \eqref{vla}:
\be\label{deco2}
\int_\C|\psi^a_z\rangle\langle\psi^a_z|d\mu_N(z)=\sum_{n=0}^{N-1} |\psiN_n\rangle\langle\psiN_n|,
\ee
where each $\psiN_n$ is proportional to the elements of the canonical basis $\{\varphi_n^N, n=1,\dots, N-1\}$, provides a decomposition of the identity on $\mathcal H_N$ endowed with a different Hilbert structure for which the $\psiN_n$s are normalized (see Section\ref{h2}). the advantage of working with the left hand side of \eqref{deco2} instead of the usual decomposition of the identity using coherent states and leading to standard T\"oplitz quantization, is the fact that $\psi_z^a$ possess an extra parameter: the density $a$. Therefore one can ``act" on $\psi_z^a$ not only by multiplication by a function $f(z)$ but by letting an operator valued function of $z$ acting on $a$. This leads to what is called in this paper $a$-T\"oplitz operators, namely operators of the form
\[
\int_\C|\psi^{\Sigma(z)a}_z\rangle\langle\psi^a_z|d\mu_N(z),
\]
where now $\Sigma(z)$ is, for each $z$, an operator acting on $\mathcal S(\R)$. The precise definition is  given in Section \ref{at} Definition \ref{atop}, and Theorem \ref{onemain} shows that matrices like $\mathcal M_1^\pm$ are $a$-T\"oplitz operators, together with their products 
%$\mathcal M_1^+\mathcal M_1^-$
 whose symbols are, at leading order, the (noncommutative) product of their symbols.
 % ones of $\mathcal M_1^+$ and $\mathcal M_1^-$.

%Note moreover that the commutator 
%\be\label{comm}
%[\mathcal M_1^+,\mathcal M_1^-]
%=
%\scriptsize
%\begin{pmatrix}
%-1&0&0&0&\dots&0\\
%0&0&0&0&\dots&0\\
%&&\dots&&&\\
%&&\dots&&&\\
%&&\dots&&&\\
%0&\dots&0&0&0&0\\
%0&\dots&0&0&0&1
%\end{pmatrix},
%\small
%\ee
%so that its symbol doesn't vanish at leading order, as for usual T\"oplitz operators, but its support is concentrated near the two singularities (poles). It comes from the parts of the two symbols of $\mathcal M_1^\pm$ localized near the poles where noncommutativity plays a full role. hence the symbol of 
%$[\mathcal M_1^+,\mathcal M_1^-]$ is of leading order too, and not of order $\hbar=\frac 1N$ as in usual semiclassical situations.

Let us remark finally that, even at the limit $\hbar=\frac\pi N=0$, the symbol of 
$\mathcal M_1^\pm$ is NOT $e^{\pm i\theta}=(z/\bar z)^{\pm\frac12}$. Traces of the noncummutative part of the symbol persist at the classical limit, as in \cite{tp}. Therefore the ``classical underlying phase-space" is not the $2$-sphere anymore, but rather a noncommutative space identified with a non commutative algebra  of such symbols playing the role of the commutative algebra of continuous functions on a standard manifold. A quick description of this space, inspired of course by noncommutative geometry \cite{ac}, is given in Section \ref{classical}.

The construction dealing with $M^\pm_1$ can be in particular generalized to matrices of the form
\be\label{mgamma}
M^N_\gamma=
\scriptsize
\begin{pmatrix}
\gamma_0(0)&\gamma_1(0)&\gamma_2(0)&\dots&\dots&\gamma_{N-1}(0)\\
\gamma_{-1}(0)&\gamma_0(1/N)&\gamma_1(1/N)&\dots&\dots&\gamma_{N-2}(1/N)\\
&&\dots&&&\\
&&\dots&&&\\
&&\dots&&&\\
\gamma_{-(N-2)}(0)&\dots&\dots&\gamma_{-1}((N-2)/N)&\gamma_0((N-2)/N)&\gamma_1((N-2)/N)\\
\gamma_{-(N-1)}(0)&\dots&\dots&\gamma_{-2}((N-3)/N)&\gamma_{-1}((N-2)/N)&\gamma_0((N-1)/N)
\end{pmatrix}.
\ee
It has been proven in \cite{bgpu,as} that such a family of matrices $M^N_\gamma$ is a T\"oplitz operator of symbol $\gamma(\tau,\theta)=\sum\limits_{k=1-N}^{N-1}\gamma_k(\tau)e^{ik\theta}$ 
%at the condition that
if and only if
\be\label{drole}
(\tau(1-\tau))^{\frac {|k|}2}\gamma_k(\tau)\in C^\infty([0,1]), k=1_N,\dots,N-1.
\ee
Condition \eqref{drole} expresses explicitly that $\gamma\in C^\infty(S_2)$.

In Section \ref{at} theorem \ref{onemain} we prove that (more general matrices than) the family $M^N_\gamma$ are $a$-T\"oplitz operators, and we compute their symbols, when \eqref{drole} is replaced by the condition
\footnote{The construction works certainly also for conditions of the type, e.g.,
\[
(\tau(1-\tau))^{\frac {\alpha|k|}2}\gamma_k(\tau)\in C^\infty([0,1]), k=1,\dots,N-1,\ 0\leq\alpha\leq 1
\]
(or even more general ones), but since we don't see any applications of these situation, we concentrate in this paper to the condition \eqref{moinsdrole}.}
\be\label{moinsdrole}
\gamma_k(\tau)\in C^\infty([0,1]), k=1-N,\dots,N-1.
\ee
Under \eqref{moinsdrole} $\gamma\notin C^\infty(S_2)$ and one has to pass form the T\"oplitz to the $a$-T\"oplitz paradigme (note that $M_1^\pm$ indeed satisfy \eqref{moinsdrole} and not \eqref{drole}).

\vskip 3cm
Let us finish this long introduction by giving the key ideas leading to the setting of our main result, Theorem \ref{main}. The reader can found in Section \ref{tqft} a very short introduction to TQFT. Larger basics on TQFT can be found in  \cite{mp}  using the same vocabulary as the present paper together with a substantial bibliography.

Combinatorial curve operators are actions of the curves on a punctured surface
$\Sigma$ on a finite
dimensional  vector space $V_r(\Sigma, c)$ indexed by a level $r$ and a coloring $c$ of the marked points taken in a set of $r$ colors.The dimension $N=N(r)$ of 
$V_r(\Sigma, c)$ will diverge as $r\to\infty$ and  $\frac1r$ can be considered as a phenomenological Planck constant $\hbar$.

In \cite{mp} we provided the construction of an explicit orthogonal basis of $V_r(\Sigma, c)$ and we conjectured that any curve operator is expressed in this basis by a matrix essentially of the form $M_\gamma$.
% where $\gamma=\gamma^r$ depends ``nicely" on $er$. 
More precisely we showed that the conjecture is true in the case where $\Sigma$ is either the punctured $2$-torus or
the $4$ times punctured sphere, Even more, we proved that (the matrix of) any curve operator belongs to the algebra generated by three matrices of the form $M^N_{\Gamma^r_0},\ M^N_{\Gamma^r_1},\ M^N_{\Gamma^r_d}$ defined in \eqref{mgamma}  where
\[
\left\{
\begin{array}{rcl}
\Gamma_0^r(\tau,\theta)&=&\gamma_0(\tau,r)\\
\Gamma_1^r(\tau,\theta)&=&2\gamma_1(\tau,r)\cos{\theta}         \\
\Gamma_d^r(\tau,\theta)&=&e^{\frac I{2r}}\gamma_1(\tau,r)\cos{(\theta+\tau)},
\end{array}
\right.
\]
for two explicit families of functions $\gamma_0,\gamma_1$.

We showed in \cite{mp} that, for ``most" values of the coloring of the marked points of $\Sigma$, the functions  $\Gamma_0^r,\Gamma_1^r,\Gamma_d^r$ are smooth functions on the sphere, and that, indeed, the corresponding curve operators are standard T\"oplitz operators. This proves also that any curve operator is T\"oplitz, by the stability result by composition  of the T\"oplitz class. Moreover the leading symbols of any curve operator happen to be the  classical trace function associated to the corresponding curve (see \cite{mp} for details).

Theorem \ref{main} of the present article express the same result
 for \textbf{any} coloring of the marked points, at the expense of replacing T\"oplitz quantization by $a$-T\"oplitz one. The only difference, unavoidably for the reason  of the change of T\"oplitz paradigm, is the fact that the $a$-T\"oplitz  leading symbol of  the curve operator is not (and cannot) the classical trace function of the curve, but we are able to compute or conjecture it out of the trace function.
% \section{The results}\label{results}
% {\bf Resultats}
% 
% 1. operateur courbe singulier : c' est un a-Toplitz regulier dans l'interieur avec formule explicite symbole principal.
% 
% 2. matrices trigonometriques = a-Toplitz
% 
% 3. a-Toplitz ; symbole, calcul symbolique, interpretation geometrique.
% 
% Ordre : 1. enonce ds intro, puis 2 puis 3.
% 
% Side orders:
% 
% a. le cas non trigonometrique
% 
% b. le genus plus grand
% 
% c. quantification geometrique revisitee.
% 
% \newpage
\section{Hilbert spaces associated to new quantizations of the sphere}\label{hilbert}
\subsection{The standard geometric quantization of the sphere}\label{canonical}
In this section we will consider the quantization of the sphere in a very down-to-earth way. See \cite{folland,mp} for more details.

Given an integer $N$, we define the space $\mathcal H_N$  of polynomials in the complex variable $z$ of order strictly less than
$N$ and set 
\be\label{basis}
\la P,Q\ra=\frac{i}{2\pi}\int_\C\frac{P(z) \ba{Q(z)}}{(1+|z|^2)^{N+1}}dzd\bar z\quad \text{and}\quad\vp^N_n (z)=\sqrt{\frac{N!}{n!(N-1-n)!}}z^n
\ee
The vectors $(\vp^N_n)_{n=0\dots N-1}$ form an orthonormal basis of $\mathcal H_N$.

By the stereographic projection $$S^2\ni (\tau,\theta)\in [0,1]\times S^1\to z=\sqrt{\frac \tau{1-\tau}}e^{i\theta}\in \C\cup\{\infty\},$$ 
The space $\mathcal H_N$ can be seen as a space of functions on the sphere (with a specific behaviour at the north pole).
Write 
\be\label{dmun}
d\mu_N=\frac{i}{2\pi}\frac{dzd\bar z}{(1+|z|^2)^{N+1}}.
\ee
 As a space of analytic functions in $L^2(\C,d\mu_N)$,  the space $\mathcal  H_N$ is closed.

For $z_0\in \C$, we define the coherent state 
\be\label{repro}\rho_{z_0}(z)=N(1+\bar z_0 z)^{N-1}.
\ee
 These vectors satisfy $\la f,\rho_{z_0}\ra=f(z_0)$ for any $f\in \mathcal H_N$ and the orthogonal projector $\pi_N:L^2(\C,d\mu_N)\to \mathcal H_N$ satisfies $(\pi_N\psi)(z)=\la \psi,\rho_z\ra$.

For $f\in \ci (S^2,\R)$ we define the (standard) T\"oplitz quantization of $f$ as the  operator 
\begin{eqnarray}
T^N[f]:\mathcal H_N&\to&\mathcal H_N\nonumber\\
T^N[f]&:=&\int_{\C}f(z)|\rho_z\rangle\langle\rho_z|d\mu_N(z)\nonumber\\
\mbox{ i.e. }T^N[f]\psi&:=&\int_\C f(z)\langle\rho_z,\psi\rangle_{\mathcal H_N} \rho_zd\mu_N(z)=\pi_N(f\psi)\mbox{ for } \psi\in\mathcal H_N.\label{deftop}
\end{eqnarray}
 A Toeplitz operator on $S^2$ is a sequence of operators $(T_N)\in \mathrm{End} (\mathcal H_N)$ such that 
there exists a sequence $f_k\in\ci(S^2,\R)$ such that for any integer $M$ the operator $R^M_N$ defined by the equation $$T_N=\sum_{k=0}^M N^{-k}T_{f_k}+R_N^M$$ is a bounded operator whose norm satisfies $||R_M||=O(N^{-M-1})$.

An easy use of the stationary phase Lemma shows that the (anti-)Wick symbol (also called Husimi function)  of $T_f$, namely $\frac {\la T_f\rho_z,\rho_z\ra}{\la\rho_z,\rho_z\ra}$ satisfies
\be\label{lap}
\frac {\la T_f\rho_z,\rho_z\ra}{\la\rho_z,\rho_z\ra}=f+ \frac 1  N \Delta_S f +O(N^{-2}),
\ee
where $\Delta_S=(1+\vert z\vert^2)^2\partial_z\partial_{\bar z}$ is the Laplacian on the sphere.

\subsection{The building vectors}\label{h1}
Let $a\in\mathcal S(\R)$, $||a||_{L^2(\R)}=1$ and $z\in\C$. We define
\be\label{def}
\psi^a_z=\int_\R a(t)e^{i\frac{\tau(z)t}\hbar}\rho_{e^{it}z}\frac{dt}{\sqrt{2\pi}}
\ee
where $\tau(z)=\frac{|z|^2}{1+|z|^2}$ and $\hbar=\frac\pi N$.

Although we won't need it in this paper, let us note that, when $z$ is far away form the origin and the point at infinity, $\psi^a_z$ is a lagrangian semiclassical distribution (WKB state) (by a similar construction as in \cite{pu}).

Since, by \eqref{repro}, $\rho_z=\sum\limits_{n=1}^{N-1}\sqrt{\frac{N!}{n!(N-1-n)!}}\bar z^n\vp^N_n$, we get that
\be\label{decomp}
\psi^a_z=\sum_{n=0}^{N-1}\widetilde a\left(\frac{\tau(z)-n\hbar}\hbar\right)\sqrt{\frac{N!}{n!(N-1-n)!}}\bar z^n\vp^N_n=\sum_{n=0}^{N-1}\widetilde a\left(\frac{\tau(z)-n\hbar}\hbar\right)\overline{\vp^N_n(z)}\vp^N_n.
\ee
where $\widetilde a$ is the Fourier transform of $a$
\[
\widetilde a(y):=\frac1{\sqrt{2\pi}}\int_\R e^{ixy}a(x)dx.
\]
\begin{remark}\label{suptildea}Note that, by \eqref{decomp}, $\psi^a_z$ depends only on the values of $\tilde a$ on $[0,N]$. Therefore one can always restrict the choice of $a$ to the functions whose Fourier transform is supported on $[0,N]$. In the sequel of this article we will always do so.
\end{remark}
\begin{lemma}
\be\label{decompa}
%\frac1N
\int_\C|\psi^a_z\rangle
%_a
\langle\psi^a_z|
%\frac i{2\pi}\frac{dzd\bar z}{(1+|z|^2)^2}
d\mu_N(z)
=\sum_{n=0}^{N-1} \cons^N_n|\vp^N_n\rangle\langle\vp^N_n|
\ee
%where the ``bra" $_a\langle\psi^a_z|$ is the adjoint of $\psi^a_z$ with respect to an $a$-dependent scalar product $\langle\cdot,\cdot\rangle_a$ to be defined later, namely $|\psi^a_z\rangle_a\langle\psi^a_z|\psi=
%\langle\psi^a_z,\psi\rangle_a\psi^a_z
%$, and
%where 
with 
\be\label{cNn}\cons^N_n=\frac{(N-1)!}{n!(N-1-n)!}\int_0^1\left|\widetilde a\left(\frac{\tau-n\hbar}\hbar\right)\right|^2\left(\frac\tau{1-\tau}\right)^n
{(1-\tau)^{N-1}}\frac{d\tau}\hbar
%(\times (\tau(1-\tau))^\frac12(N-1)^\frac32?)
.
\ee
Moreover, as $N-1=\frac1\hbar\to\infty$
 %and $0<n\hbar<1$,
\be\label{casymp}
\cons^N_n
%\to ||a||_{L^2(\R)}
=1+O(\frac1N) ,\ \ \ 0<n\hbar<1
%\frac 1{N-1}
.
\ee
\be\label{casymp0}
\cons^N_n\sim \frac1{n!}\int_0^\infty |\widetilde a(\lambda-n)|^2\lambda^ne^{-\lambda}\sqrt{2\pi\lambda}d\lambda ,\ \ \ 0\sim n\hbar
%\frac 1{N-1}
.
\ee

\be\label{casymp1}
\cons^N_n\sim C^N_{N-1-n},\ \ \  n\hbar\sim1
%\frac 1{N-1}
.
\ee
\end{lemma}
\begin{proof}
Deriving \eqref{decompa} is a straightforward calculus after \eqref{decomp}.

By the asymptotic formula for the binomial we get that, as $N,n\to\infty$,
\[
\frac{(N-1)!}{n!(N-1-n)!}\sim
\left(\frac{\frac n{N-1}}{1-\frac n{N-1}}\right)^{-n}
\left(1-\frac n{N-1}\right)^{(N-1)}.
\]
Moreover since $0<n\hbar<1$ we get since $\widetilde a$ is fast decreasing at infinity,
\[
\int_0^1\left|\widetilde a\left(\frac{\tau-n\hbar}\hbar\right)\right|^2\left(\frac\tau{1-\tau}\right)^n
{(1-\tau)^{N-1}}\frac{d\tau}\hbar\sim
\int_{-\infty}^{+\infty}\left|\widetilde a\left(\frac{\tau-n\hbar}\hbar\right)\right|^2\left(\frac\tau{1-\tau}\right)^n
{(1-\tau)^{N-1}}\frac{d\tau}\hbar
\]
and $\frac1\hbar\left|\widetilde a\left(\frac{\tau-n\hbar}\hbar\right)\right|^2\to||a||_{L^2(\R)}\delta(\tau-n\hbar)$ as $\hbar=\frac1{N-1}\to 0$. Therefore we get \eqref{casymp}.

\end{proof}
\begin{definition}\label{defpsi}
\[
\psiN_n:=\sqrt{\cons^N_n}\vp^N_n.
\]
\end{definition}
This definition is motivated by \eqref{decompa} which actually reads 
\be\label{vla}
\int_\C|\psi^a_z\rangle\langle\psi^a_z|d\mu_N(z)=\sum_{n=0}^{N-1} |\psiN_n\rangle\langle\psiN_n|.
%={\bf 1}_{\mathcal H_N^a},
\ee
This leads to the following equality:
\be\label{vlaa}
\int_\C|\psi^a_z\rangle_a\langle\psi^a_z|d\mu_N(z)={\bf 1}_{\mathcal H_N^a},
\ee
where $H_N^a$ is the same space of polynomials as $H_N$ but now endowed with the renormalized scalar product $\langle\cdot,\cdot\rangle_a$  fixed by
\be\label{oufoufouf}
\langle\psi^N_m,\psi^N_n\rangle_a=\delta_{m,n},
\ee
and
\be\label{diraca}
%\left(
|\psi^a_z\rangle_a\langle\psi^a_z|
%\right)
\psi:=\langle \psi_z^a,\psi\rangle_a\psi_z^a,\ \psi\in \mathcal H_N^a.
\ee
\subsection{The Hilbert structure}\label{h2}
%We consider on $\mathcal H_N$ t
The Hilbert scalar product on $\mathcal H^a_N$ is obtained out of  \eqref{oufoufouf} by bi-linearity. Since any polynomial $f$ satisfies 
$$
f=\sum\limits_0^{N-1}\langle\vp^N_n,f\rangle\vp^N_b=\frac1{C^N_n}
\langle\psi^N_n,f\rangle\psi^N_b, 
$$
we get
\[
\langle f,g\rangle_a:=\sum_{n=0}^{N-1}\frac1{(C^N_n)^2}\langle f,\psiN_n\rangle\langle\psiN_n, g\rangle=
=\sum_{n=0}^{N-1}\frac1{C^N_n}\langle f,\vp^N_n\rangle\langle\\vp^N_n, g\rangle.
\]
%and let us denote by $\mathcal H^a_N$ the corresponding Hilbert space.

Note that $\langle,\rangle_a$ is not given by an integral kernel. But if we ``change" of representation and define $F(z):=\langle\psi^a_z,f\rangle,\ G(z):=\langle\psi^a_z,g\rangle$ then, by \eqref{vla} we have
\[
\langle f,g\rangle_a=\langle F,G\rangle=\int_\C\bar{F(z)}G(z)d\mu_N(z).
\]
Let us remark finally that
\[
F(z)=\int_\R a(t)e^{i\frac{\tau(z)t}\hbar}f(e^{it}z)\frac{dt}{\sqrt{2\pi}}
\]
and
\[
f=\int_\C F(z)|\psi^a_z\rangle dz,\]
 namely 
 \[ f(z')=\int_\C F(z)\psi^a_{z}(z') dz.
\]
%\subsection{Geometric interpretation}\label{h3}
%\newpage
\section{Singular quantization}\label{sq}
This section is the heart for the present paper. We will first show how operators on $\mathcal H^a_N$  defined as matrices on the basis $\{\psi^M_n,\}$ act on the building operators $\psi^a_z$ by action on $a$ (section \ref{action}). This will allow us, in section \ref{symbol}, to assign to each of these matrices  symbols whose symbolic calculus is studied in section \ref{calculus}. This will lead us finally to section \ref{at} where we define the $a$-T\"oplitz quantization.
\subsection{A toy model case}\label{toy}
Let us consider the $N\times N$ matrix
\begin{eqnarray}
%\unun_1&=&
\  \begin{pmatrix}
0&1&0&0&\dots&0\\
1&0&1&0&\dots&0\\
&&\dots&&&\\
&&\dots&&&\\
&&\dots&&&\\
0&\dots&0&1&0&1\\
0&\dots&0&0&1&0
\end{pmatrix}
%\nonumber
%\\
&=&
\begin{pmatrix}
0&0&0&0&\dots&0\\
1&0&0&0&\dots&0\\
&&\dots&&&\\
&&\dots&&&\\
&&\dots&&&\\
0&\dots&0&1&0&0\\
0&\dots&0&0&1&0
\end{pmatrix}
+
\begin{pmatrix}
0&1&0&0&\dots&0\\
0&0&1&0&\dots&0\\
&&\dots&&&\\
&&\dots&&&\\
&&\dots&&&\\
0&\dots&0&0&0&1\\
0&\dots&0&0&0&0
\end{pmatrix}\nonumber\\
=:\hskip 2cmM_1\hskip 1.7cm &=:&
\hskip 2cmM_1^+\hskip 1.7cm + \hskip 2cm M_1^-\nonumber
\end{eqnarray}
and let us consider the operator $\mathcal\unun_1=\mathcal\unun_1^++\mathcal\unun_1^-$ on $\mathcal H^a_N$ whose matrix on the orthonormal basis $\{\psiN_n,n=0,\dots,N-1\}$ is $M$. That is
\[
\left\{\begin{array}{lcll}
\mathcal\unun_1^\pm\psiN_0&=&\frac{
%1^
1\pm
1}2&\psiN_1\\
\mathcal\unun_1^\pm\psiN_i&=&&\psiN_{i\pm 1},\ \ \ 1\leq i\leq N-2\\
\mathcal\unun_1^\pm\psiN_{N-1}&=&\frac{
%1^
1\mp
1}2&\psiN_{N-2}
\end{array}\right.
\]
%\[
%\begin{pmatrix}
%a&c\\b&b
%\end{pmatrix}
\begin{proposition}\label{thefact}
\[
\mathcal\unun_1\psi^a_x=\psi^{\Sigma_1(z)a}_z
\]
where the operator $\Sigma_1(z)$ is given by \eqref{sigma} below.
%with
%\[
%b_{\bar z}=\Sigma(\bar z)a,\ \ \Sigma(\bar z)=
%\]
\end{proposition}
\begin{proof}
By \eqref{decomp} we get that, calling $\conss^N_n=(\cons^N_n)^{-\frac12}$  (once again $\hbar=\frac1{N-1}$),
\[
\psi^a_z=\sum_{n=0}^{N-2}\widetilde a\left(\frac{\tau(z)-n\hbar}\hbar\right)
%\overline{\vp^N_n(z)}
\sqrt N\sqrt{\binom{N-1}{n}}\bar z^nD^N_n
\psiN_{n}.
\]
Therefore
\begin{eqnarray}
\mathcal\unun_1\psi^a_z\nonumber\\
=\widetilde a\left(\frac{\tau(z)}\hbar\right)\sqrt ND^N_0\psiN_1\nonumber\\
+
\sum_{n=1}^{N-2}\widetilde a\left(\frac{\tau(z)-n\hbar}\hbar\right)
%\overline{\vp^N_n(z)}
\sqrt N\sqrt{\binom{N-1}{n}}\bar z^nD^N_n
(\psiN_{n-1}+\psiN_{n+1})\nonumber\\
+
\widetilde a\left(\frac{\tau(z)-1}\hbar\right)\sqrt N\bar z^{N-1}D^N_{N-1}\psiN_{N-2}\nonumber\\
=
\widetilde a\left(\frac{\tau(z)}\hbar\right)\sqrt ND^N_0\psiN_1
+
\widetilde a\left(\frac{\tau(z)-1}\hbar\right)\sqrt N\bar z^{N-1}D^N_{N-1}\psiN_{N-2}\nonumber\\
+
\sum_{n=0}^{N-3}
\widetilde a\left(\frac{\tau(z)-(n+1)\hbar}\hbar\right)
%\overline{\psiN_n(z)}
\sqrt N\sqrt{\binom{N-1}{n+1}}\bar z^{n+1}D^N_{n+1}
\psiN_{n}\nonumber\\
+
\sum_{n=2}^{N-1}
\widetilde a\left(\frac{\tau(z)-(n-1)\hbar}\hbar\right)
%\overline{\vp^N_n(z)}
\sqrt N\sqrt{\binom{N-1}{n-1}}\bar z^{n-1}D^N_{n-1}
\psiN_{n}\nonumber\\
=
\sum_{n=0}^{N-2}
\widetilde a\left(\frac{\tau(z)-(n+1)\hbar}\hbar\right)
%\overline{\vp^N_n(z)}
\sqrt N\sqrt{\binom{N-1}{n+1}}
\bar z^{n+1}D^N_{n+1}\psiN_{n}\label{nord}\\
+
\sum_{n=1}^{N-1}
\widetilde a\left(\frac{\tau(z)-(n-1)\hbar}\hbar\right)
%\overline{\vp^N_n(z)}
\sqrt N\sqrt{\binom{N-1}{n-1}}\bar z^{n-1}
D^N_{n-1}\psiN_{n}\label{sud}
\end{eqnarray}
Let us consider the sum in \eqref{sud}. On can write it as
\begin{eqnarray}
\psi_{{sud}}&=&\sum_{n=0}^{N-1}\frac1{\bar z}\mu(n)\widetilde a\left(\frac{\tau(z)-(n-1)\hbar}\hbar\right)
\sqrt N\sqrt{\binom{N-1}{n}}\bar z^nD^N_n
\psiN_n\nonumber\\
&=&
\sum_{n=0}^{N-1}\frac1{\bar z}\mu(n)\widetilde a\left(\frac{\tau(z)-(n-1)\hbar}\hbar\right)
\bar{\varphi_n(z)}\varphi_n
%(z)
\nonumber
\end{eqnarray}
with
\be\label{ddefmu}
\mu(n)=\left\{\begin{array}{l}
\sqrt{
%\frac{\binom{N-1}{n-1}}{\binom{N-1}{n}}
\frac n{N-n}
}
,\ n>0\\
0,\ n=0\end{array}\right.
\ee
We get that
\[
\psi_{{sud}}=\psi^{b_{{sud}}}_z
\]
where
\[
b_{{sud}}=\sqrt{\frac{C^N_{\cdot}}{C^N_{\cdot-1}}}\mu\left(\frac{\tau(z)}\hbar
%+\frac12
-i\partial_x\right)\frac{e^{ix}}{\bar z}a=\Sigma_1^-a.
\]
%\[
%b_{{sud}}(x)=\frac1{\bar z}\mu(-i\partial_x)\big(e^{-i\frac{\tau(z)x}\hbar}e^{-ix}a(-x)\big).
%\]
Here we have denote by $\frac{C^N_{\cdot}}{C^N_{\cdot-1}}$  the function defined out of \eqref{cNn} by
\be\label{fraccncn}
\frac{C^N_{\cdot}}{C^N_{\cdot-1}}:\ \xi\in]0,N[\to 
\sqrt{\frac{\xi}{N-\xi}}
\frac{\int_0^1\left|\widetilde a\left(\frac{\tau-\xi\hbar}\hbar\right)\right|^2\left(\frac\tau{1-\tau}\right)^\xi
{(1-\tau)^{N-1}}\frac{d\tau}\hbar}
{\int_0^1\left|\widetilde a\left(\frac{\tau-(\xi-1)\hbar}\hbar\right)\right|^2\left(\frac\tau{1-\tau}\right)^\xi
{(1-\tau)^{N-1}}\frac{d\tau}\hbar},
\ee
and $\sqrt{\frac{C^N_{\cdot}}{C^N_{\cdot-1}}}\mu$ is meant as the product of the two functions, i.e. $\sqrt{\frac{C^N_{\cdot}}{C^N_{\cdot-1}}}\mu(\xi)=
\sqrt{\frac{C^N_{\cdot}}{C^N_{\cdot-1}}}(\xi)\mu(\xi)$ using \eqref{fraccncn}. Note finally that, by the band limited hypothesis on $a$ in Remark \ref{suptildea}, $b_{{sud}}$ is well defined.

Similarly we get that the sum in \eqref{nord} is 
\[
\psi_{{nord}}=\sum_{n=0}^{N-1}{\bar z}\nu(n)\widetilde a\left(\frac{\tau(z)-(n+1)\hbar}\hbar\right)
\sqrt N\sqrt{\binom{N-1}{n}}\bar z^nD^N_n
\psiN_n
\]
with
\[
\nu(n)=\left\{\begin{array}{l}
\sqrt{
%\binom{N-1}{n+1}
\frac{N-n-1}{n+1}
}
,\ n<N-1\\
0,\ n=N-1\end{array}\right.
\]
So
\[
\psi_{{nord}}=\psi^{b_{{nord}}}_z
\]
where
\[
b_{{nord}}=\sqrt{\frac{C^N_{\cdot}}{C^N_{\cdot+1}}}\nu^N\left(\frac{\tau(z)}\hbar
%-\frac12
-i\partial_x\right){\bar z}e^{-ix}a=\Sigma_1^+a.
\]
%\[
%b_{{nord}}(x)={\bar z}\nu(-i\partial_x)\big(e^{-i\frac{\tau(z)x}\hbar}e^{+ix}a(-x)\big).
%\]
We define
\be\label{sigma}
\Sigma_1(z)=
%\frac1{\bar z}\mu(-i\partial_x)e^{-i\frac{\tau(z)x}\hbar}e^{-ix}I+{\bar z}\nu(-i\partial_x)e^{-i\frac{\tau(z)x}\hbar}e^{+ix}I
\sqrt{\frac{C^N_{\cdot}}{C^N_{\cdot+1}}}\mu^N\left(\frac{\tau(z)}\hbar
%+\frac12
-i\partial_x\right)\frac{e^{ix}}{\bar z}+\sqrt{\frac{C^N_\cdot}{C^N_{\cdot-1}}}\nu^N\left(\frac{\tau(z)}\hbar
%-\frac12
-i\partial_x\right){\bar z}e^{-ix}.
\ee 
\[
=
\sqrt{\frac{C^N_{\cdot}}{C^N_{\cdot+1}}}\Sigma_1^+(z)+
\sqrt{\frac{C^N_\cdot}{C^N_{\cdot-1}}}\Sigma_1^-(z)\hskip 4.3cm
\]
where 
\be\label{munu}
\mu^N=\chi_{[\frac12,N-\frac12]}\mu,\ \nu^N=\chi_{[-\frac12,N-\frac32]}\nu,
\ee
$\chi\in C^\infty(\R)$ satisfies
\be\label{qui}
\chi_{[a,b]}(\xi)=\left\{\begin{array}{ccl}
0&\mbox{ if }&\xi\leq a\\
\chi'(\xi)>0&\mbox{ if }&a<\xi<a+\frac12\\
1&\mbox{ if }&a+\frac12\leq \xi\leq b=\frac12\\
\chi'(\xi)<0&\mbox{ if }&b-\frac12<\xi<b\\
0&\mbox{ if }&b\leq \xi
\end{array}
\right.
\ee
and 
$\mu^N\left(\frac{\tau(z)}\hbar
%+\frac12
-i\partial_x\right)$ and 
$\nu^N\left(\frac{\tau(z)}\hbar
%+\frac12
-i\partial_x\right)$ are defined by the spectral theorem applied to the operator $-i\partial_x$ acting on $L^2(\mathbb R)$. Moreover $\psi^{b_{sud}}_z,\ \psi^{b_{nord}}_z$ depend only on $\mu^N(\frac\tau\hbar-n),\ \nu^N(\frac\tau\hbar-n)$, so that they depend only on the properties \eqref{qui} of $\chi$.
%where $I$ is the parity operator $Ia(x)=a(-x)$.
\end{proof}
%\begin{color}{red}
%blabla sur le lien avec le cas regulier, symbole Toplitz etc.
%\end{color}
{\bf In order to make the notations a bit lighter, we will skip the over-script $N$ in $\mu^N$ and $\nu^N$ in the sequel of the paper.}
\subsection{(General) trigonometric matrices}\label{trigomat}

Let
\be\label{sigmapm}
\Sigma_1^+(z)=\Sigma_1^+=
\mu^N\left(\frac{\tau(z)}\hbar
%+\frac12
-i\partial_x\right)\frac{e^{ix}}{\bar z},\ \Sigma_1^-(z)=\Sigma_1^-=\nu^N\left(\frac{\tau(z)}\hbar
%-\frac12
-i\partial_x\right){\bar z}e^{-ix}
\ee
as defined by \eqref{sigma}.

We get easily the following result.
\begin{lemma}\label{cefutdur}
\begin{eqnarray}
\Sigma_1^+\Sigma_1^-&=&\chi_{[+\frac12,N-\frac12]}\left(\frac{\tau(z)}\hbar-i\partial_x\right)\label{+-}\\
\Sigma_1^-\Sigma_1^+&=&\chi_{[-\frac12,N-\frac32]}\left(\frac{\tau(z)}\hbar-i\partial_x\right)\label{-+}\\
{[}\Sigma_1^-,\Sigma_1^+{]} 
&=&\bar\chi\left(\frac{\tau(z)}\hbar-i\partial_x\right)
%_{-\frac12,+\frac12}
\label{-++-}
\end{eqnarray}
%where $\bar\chi=\chi_{[\frac12,N-\frac32]}-\chi_{[-\frac12,N-\frac12]}$, that is
%\[
%\bar\chi(\xi)=\left\{
%\begin{array}{ccl}
%0&\mbox{ if }&\xi\leq -\frac12\\
%0<\bar\chi'&\mbox{ if }&-\frac12<\xi<0\\
%1&\mbox{ if }&\xi=0\\
%\bar\chi'<0&\mbox{ if }&0<\xi<\frac12\\
%0&\mbox{ if }&\frac12\leq \xi\leq N-\frac32\\
%\bar\chi'<0&\mbox{ if }&N-\frac32< \xi<N-1\\
%1&\mbox{ if }&\xi=N-1\\
%0<\bar\chi'&\mbox{ if }&N-1\leq  \xi\leq N-\frac12\\
%0&\mbox{ if }&N-\frac12\leq\xi
%\end{array}\right.
%\]
is
\[
\bar\chi(\xi)=\left\{
\begin{array}{ccl}
0&\mbox{ if }&\xi\leq -\frac12\\
0<\bar\chi'&\mbox{ if }&-\frac12<\xi<0\\
1&\mbox{ if }&0\leq\xi\leq\frac12\\
\bar\chi'<0&\mbox{ if }&\frac12<\xi<1\\
0&\mbox{ if }&1\leq \xi\leq N-2\\
\bar\chi'<0&\mbox{ if }&N-2< \xi<N-\frac32\\
-1&\mbox{ if }&N-\frac32\leq\xi\leq N-1\\
0<\bar\chi'&\mbox{ if }&N-1\leq  \xi\leq N-\frac12\\
0&\mbox{ if }&N-\frac12\leq\xi
\end{array}\right.
\]
\end{lemma}

%\newpage
%\begin{proposition}\label{regas}
%Let us fix $z\neq 0,\infty$ and define 
%$e^{i\theta(z)}=\sqrt{\frac z{\bar z}}$. Then
%\[
%\Sigma_1(z)
%%a(x)
%\sim2\cos{(x+\theta(z))}
%%a(x)
%\mbox{ as }N\to\infty.
%\]
%Let us fix now $z$ such that $\tau(z)\sim \sqrt\hbar$ 
%(or $\tau(z)\sim1-\sqrt\hbar$ ?)
%%, and $\frac{D^N_n}{D^N_{n-1}}=1+O(\hbar)$
%. Then
%\[
%\Sigma_1(z)\sim e^{i(x+\theta(z))}
%\sqrt[N]{1-i\frac\hbar{|z|^2}\partial_x}\frac{D^N_{D_x^+}}{D^N_{D_x^+1}}
%+e^{-i(x+\theta(z))}
%\frac1{\sqrt[N]{1-i\frac\hbar{|z|^2}\partial_x}}\frac{D^N_{D_x^+}}{D^N_{D_x^-1}}
%\]
%where $\sqrt[N]A=\sqrt {Fourier^{-1}\circ(FourierAFourier^{-1})|_{0\leq \xi\leq N-1}\circ Fourier}$ 
%
%$D^+=Fourier(\xi I_{\xi\geq 0})Fourier^{-1}$.
%\end{proposition}
\begin{remark}
When $z$ is far away from the origin or the infinity, the ``symbol" $\Sigma$ at $z$ is just an operator of multiplication, therefore ``commutative". And it is as expected equal to, basically, $2\cos \theta$. But $2\cos\theta$ is not regular at the two poles, and the trace of this singularity is the fact that $\Sigma(z)$ becomes a non-local operator when $z$ close to the poles, coming from the fact that the vector field expressed by the transport equation becomes infinite.
\end{remark}
\begin{remark}\label{defl}
By \eqref{decompa} we have that
\[
\int_\C|\psi^a_z\rangle\langle\psi^a_z|
%\frac i{2\pi}\frac{dzd\bar z}{(1+|z|^2)^2}
d\mu_N(z)=C^N_{L}
\]
where $L\varphi_n=n\varphi_n$. Therefore we could also look at matrices acting on $\mathcal H^N$ instead of $\mathcal H^N_a$ by conjugation by $C^N_{L}$. But this doesn't give anything interesting for symbols.
\end{remark}
%rev
Let us generalize this to the situation where $M$ has the form, for $\alpha\in C^\infty(]0,1[)\cap L^\infty([0,1])$\footnote{by this we mean that $\alpha$ is bounded on $[0,1]$ and $C^\infty$ on any open subset of $[0,1]$.},
\[
\unun_{1,\alpha}=\begin{pmatrix}
0&\alpha(\hbar)&0&0&\dots&0\\
\alpha(\hbar)&0&\alpha(2\hbar)&0&\dots&0\\
&&\dots&&&\\
&&\dots&&&\\
&&\dots&&&\\
0&\dots&0&\alpha((N-3)\hbar)&0&\alpha((N-2)\hbar)\\
0&\dots&0&0&\alpha((N-2)\hbar)&0
\end{pmatrix}
\]
The operator $\mathcal\unun_{1,\alpha}$ on $\mathcal H^a_N$ whose matrix on the orthonormal basis $\{\psiN_n,n=0,\dots,N-1\}$ is $M_{1,\alpha}$ becomes
\[
\left\{\begin{array}{lcl}
\mathcal\unun_{1,\alpha}\psiN_0&=&\alpha(\hbar)\psiN_1\\
\mathcal\unun_{1,\alpha}\psiN_i&=&\alpha((i-1)\hbar)\psiN_{i-1}+\alpha((i+1)\hbar)\psiN_{i+1},\ \ \ 1\leq i\leq N-2\\
\mathcal\unun_{1,\alpha}\psiN_{N-1}&=&\alpha((N-2)\hbar)\psiN_{N-2}
\end{array}\right.
\]
The same type of computations contained in the proof of Proposition \ref{thefact} provides, thanks to Lemma \ref{cefutdur}, the proofs of the next Propositions \ref{thefactalpha}, \ref{thefactbeta} and \ref{thefactgamma} below.
\begin{proposition}\label{thefactalpha}
\[
\mathcal\unun_{1,\alpha}\psi^a_x=\psi^{\Sigma_{1,\alpha}(z)a}_z
\]
where 
\be\label{sigmalpha}
\Sigma_{1,\alpha}(z)=\frac{e^{ix}}{\bar z}(\alpha(\hbar\cdot)\mu)\left(\frac{\tau(z)}\hbar+\frac12-i\partial_x\right)+{\bar z}e^{-ix}(\alpha(\hbar\cdot)\nu)\left(\frac{\tau(z)}\hbar-\frac12-i\partial_x\right).
\ee
%with
%\[
%b_{\bar z}=\Sigma(\bar z)a,\ \ \Sigma(\bar z)=
%\]
In particular 
%rev if $\alpha(\tau)=\sqrt{\tau(1-\tau)}\beta(\tau)$ where $\beta\in C^\infty([0,1])$, 
if ${\sqrt{\tau(1-\tau)}}{\alpha(\tau)}\in C^\infty([0,1])$, 
so that $\alpha(\tau)e^{i\theta}\in C^\infty(S^2)$, then, for all $z\in S^2$, 
$\Sigma_{1,\alpha}(z)\sim2\alpha(\tau(z))\cos{2(x+\theta(z))}$  as $N\to\infty$. 

Otherwise, 
%rev
%if $\alpha\in C^\infty([0,1])$, 
this last asymptotic equality is valid only for $z$ away from the two poles.
\end{proposition}
\vskip 1cm
Let now, again for $\beta\in C^\infty(]0,1[)\cap L^\infty([0,1])$,  
\[
\unun_{2,\beta}=\begin{pmatrix}
0&0&\beta(\hbar)&0&0&\dots&0\\
0&0&0&\beta(2\hbar)&0&\dots&0\\
\beta(2\hbar)&0&0&0&\beta(3\hbar)&\dots&0\\
%&&\dots&&&\\
&&\dots&&&\\
\dots&0&\beta((N-4)\hbar)&0&0&0&\beta((N-3)\hbar)\\
0&\dots&0&\beta((N-3)\hbar)&0&0&0\\
0&\dots&0&0&\alpha((N-2)\hbar)&0&0
\end{pmatrix}
\]
The operator $\mathcal\unun_{2,\beta}$ on $\mathcal H^a_N$ whose matrix on the orthonormal basis $\{\psiN_n,n=0,\dots,N-1\}$ is $M_{2,\beta}$ becomes
\[
\left\{\begin{array}{lcl}
\mathcal\unun_{2,\beta}\psiN_0&=&\alpha(\hbar)\psiN_2\\
\mathcal\unun_{2,\beta}\psiN_1&=&\beta(2\hbar)\psiN_3\\
\mathcal\unun_{2,\beta}\psiN_i&=&\beta((i-2)\hbar)\psiN_{i-2}+\alpha((i+2)\hbar)\psiN_{i+2},\ \ \ 2\leq i\leq N-3\\
\mathcal\unun_{2,\beta}\psiN_{N-2}&=&\alpha((N-4)\hbar)\psiN_{N-4}\\
\mathcal\unun_{2,\beta}\psiN_{N-1}&=&\alpha((N-3)\hbar)\psiN_{N-3}
\end{array}\right.
\]
\begin{proposition}\label{thefactbeta}
\[
\mathcal\unun_{2,\beta}\psi^a_x=\psi^{\Sigma_{2,\beta}(z)a}_z
\]
where 
\be\label{sigmabeta}
\Sigma_{2,\beta}(z)=\frac{e^{2ix}}{\bar z^2}(\beta(\hbar\cdot)\mu_2)\left(\frac{\tau(z)}\hbar+\frac32-i\partial_x\right)+{\bar z}^2e^{-i2x}(\beta(\hbar\cdot)\nu_2)\left(\frac{\tau(z)}\hbar-\frac32-i\partial_x\right).
\ee
with
\[
\mu_2(n)=\sqrt{\frac{(n)(n-1)}{(N+1-n)(N-n)}}\frac{D^N_{n+2}}{D^N_n}
\mbox{ and }\nu_2(n)=\sqrt{\frac{(N-1-n)(N-2-n)}{(n+2)(n+1)}}\frac{D^N_{n-2}}{D^N_n}
\]
And again  if  ${\tau(1-\tau)}{\beta(\tau)}\in C^\infty([0,1])$, so that $\beta(\tau)e^{i2\theta}\in C^\infty(S^2)$, then, for all $z\in S^2$, 
$\Sigma_{2,\beta}(z)\sim2\beta(\tau(z))\cos{2(x+\theta(z))}$  as $N\to\infty$. 

Otherwise,  this last asymptotic equality is valid only for $z$ away from the two poles.
\end{proposition}
\vskip 1cm
Let us finally remark that when $M_{0\gamma}$ is diagonal with diagonal matrix elements $\gamma(i\hbar)$, then $\Sigma_{0\gamma}=\gamma(\tau(z))\mbox{Id}$, where $\mbox{Id}$ is the identity on $L^2(\R)$.

\subsection{Action of a general matrix}\label{action}

For $k=-(N-1),\dots, N-1$, let us call $N_{k;\gamma_k}$ the matrix with  non zero coefficients lying only on the $k$th diagonal and being  equal to $\gamma_k(j)=\gamma_k(j), k\leq j\leq N-i-k$. That is to say:
\[
N_{k;\gamma_k}={(M^+_1)}^kM_{0,\gamma_k}. 
\]
Let moreover
\begin{eqnarray}
\mu_k(n)
&=&
\prod_{j=0}^{k-1}\mu(n-j)
=\sqrt{\binom{n}{k}\binom{N-n+k-1}{k}^{-1}}
\prod_{j=0}^{k-1}\chi_{[\frac12,N-\frac12]}(n-j)\hskip 1.9cm k>0
\nonumber\\
\mu_0(n)&=&1\nonumber\\
 \mu_k(n)&=&
\prod_{j=k-1}^{0}\nu(n+j)
=\sqrt{\binom{N-1-n}{k}\binom{n+k}{k}^{-1}}
\prod_{j=k-1}^{0}\chi_{[-\frac12,N-\frac12]}(n+j)\hskip 1cm k<0\nonumber
\end{eqnarray}
The same arguments as in the proofs of Propositions \ref{thefact}, \ref{thefactalpha} and \ref{thefactbeta} leads easily to the following more general result.
\begin{proposition}\label{thefactgamma}
\[
\mathcal\unun_{k,\gamma_k}\psi^a_x=\psi^{\Sigma_{k,\gamma_k}(z)a}_z
\]
where 
\begin{eqnarray}\label{sigmabeta}
\Sigma_{k,\gamma_k}(z)&=&\frac{e^{ikx}}{\bar z^k}(\gamma_k(\hbar\cdot)\mu_k)\left(\frac{\tau(z)}\hbar-i\partial_x\right)\sqrt{\frac{C^N_{-i\partial_x+k}}{C^N_{-i\partial_x}}}\nonumber\\
&+&{\bar z}^ke^{-ikx}(\gamma_k(\hbar\cdot)\nu_k)\left(\frac{\tau(z)}\hbar--i\partial_x\right)
\sqrt{\frac{C^N_{\frac{\tau(z)}\hbar--i\partial_x-k}}{C^N_{\frac{\tau(z)}\hbar--i\partial_x}}}.
\end{eqnarray}
%with
%\[
%\mu_2(n)=\sqrt{\frac{(n)(n-1)}{(N+1-n)(N-n)}}\frac{D^N_{n+2}}{D^N_n}
%\mbox{ and }\nu_2(n)=\sqrt{\frac{(N-1-n)(N-2-n)}{(m+2)(n+1)}}\frac{D^N_{n-2}}{D^N_n}
%\]
And again  if  $(\tau(1-\tau))^{\frac {|k|}2}\gamma_k(\tau)\in C^\infty([0,1])$, so that $\gamma_k(\tau)e^{i2\theta}\in C^\infty(S^2)$, then, for all $z\in S^2$, 
$\Sigma_{k,\gamma_k}(z)\sim2\gamma_k(\tau(z))\cos{k(x+\theta(z))}$  as $N\to\infty$. 

Otherwise,  this last asymptotic equality is valid only for $z$ away from the two poles.
\end{proposition}
\vskip 1cm
\subsection{Symbol}\label{symbol}
\ 

%\begin{color}{red}
%1. Le symbole va etre une fonction a valeur operatuer sur $L^2$, fonction de $\hbar$ et de $z/\hbar$.  Mias que l' on ne regarde que comme agissant sur des fonctions non-oscillantes, donc restreint a la section nulle.
%
%2. Quand $\hbar\to 0,\ z/\hbar>>1$, alors le symbole est une fonction(potentiel) periodique, dependante de $z$.Par exemple pour $z\sim1$. Donc la multiplication des symbole est commutative dans ce regime.
%
%3. Quand $\hbar\to 0,\ z/\hbar$ reste ``borne", alors le symbole est un pseudodiff semiclassique avec constante de Planck ``effective" egale a $\hbar/z$. La multiplication est alors noncommutative.
%
%4. Quand $\hbar\to 0,\ z/\hbar\to 0$, alors le symbole explose en $z\sim 0$, mais l' integrale definissant l' operateur total converge grace auz cut-offs.
%
%5. Dans ce cas-la, le symbole se factorise en un pole en $z$ et un sous-symbole du meme ordre.Comme on fait agir sur la section nulle, cette contribution disparait a la limite $\hbar\to 0$.
%\end{color}

Let us first remark the following co-cycle property.
\begin{lemma}\label{coc}
\[
\sqrt{\frac{C^N_{-i\partial_x+k'}}{C^N_{-i\partial_x}}}e^{ikx}
\sqrt{\frac{C^N_{-i\partial_x+k}}{C^N_{-i\partial_x}}}
=
e^{ikx}
\sqrt{\frac{C^N_{-i\partial_x+k'+k}}{C^N_{-i\partial_x}}}
\]
so that
\begin{eqnarray}
\frac{e^{ik'x}}{\bar z^{k'}}(\gamma_{k'}(\hbar\cdot)\mu_{k'})\left(\frac{\tau(z)}\hbar-i\partial_x\right)\sqrt{\frac{C^N_{-i\partial_x+{k'}}}{C^N_{-i\partial_x}}}
\frac{e^{ikx}}{\bar z^k}(\gamma_k(\hbar\cdot)\mu_k)\left(\frac{\tau(z)}\hbar-i\partial_x\right)\sqrt{\frac{C^N_{-i\partial_x+k}}{C^N_{-i\partial_x}}}\nonumber\\
=
\frac{e^{ik'x}}{\bar z^{k'}}(\gamma_{k'}(\hbar\cdot)\mu_{k'})\left(\frac{\tau(z)}\hbar-i\partial_x\right)
\frac{e^{ikx}}{\bar z^k}(\gamma_k(\hbar\cdot)\mu_k)\left(\frac{\tau(z)}\hbar-i\partial_x\right)
\sqrt{\frac{C^N_{-i\partial_x+k+k'}}{C^N_{-i\partial_x}}}\nonumber
\end{eqnarray}
\end{lemma}
%For $k=-(N-1),\dots, N-1$, let us call $N_{k;\gamma_k}$ the matrix with  non zero coefficients lying only on the $k$th diagonal and being  equal to $\gamma_k(j)=\gamma_k(j), k\leq j\leq N-i-k$. That is to say:
%\[
%N_{k;\gamma_k}={(M^+_1)}^kM_{0,\gamma_k}. 
%\]
Let us denote by $\mathcal N_{k;\gamma_k}$ the operator whose matrix on the basis $\{\psi^N_n,\ n=0\dots N-1\}$ is $N_{k;\gamma_k}$.

We define the symbol of $\mathcal N_{k;\gamma_k}$ at the point $z$ as the operator
\be\label{whyes}
\widetilde\sigma_{{k;\gamma_k}}(z):=
\frac{e^{ikx}}{\bar z^k}
(\gamma_k(\hbar\cdot)\mu_k)\left(\frac{\tau(z)}\hbar-i\partial_x\right)
=\frac{e^{ikx}}{\bar z^k}\gamma_k(\tau(z)-i\hbar\partial_x)\mu_k\left(\frac{\tau(z)}\hbar-i\partial_x\right)
\ee
acting on $L^2(\R)$.
\begin{definition}\label{defsym}
Let $\gamma(\tau,\theta)=\sum\limits_{k=-K}^{K}\gamma_k(\tau)e^{ik\theta}$ be a trigonometric function on the sphere with each $\gamma_k\in C^\infty(]0,1[)\cap L^\infty([0,1])$.

Let 
\be\label{defngamma}
N_\gamma=\sum\limits_{-(N-1)}^{N-1}N_{k;\gamma_k}\mbox{ where 
}(N_{k;\gamma_k})_{ij}=\delta_{j,i+k}
\gamma_k((k-\frac{(-1)^k-1}2)\hbar)
\ee
 and $\mathcal N_\gamma$ the operator whose matrix on the basis $\{\psi^N_n\}$ is $N_\gamma$.

%\begin{definition}\label{defsym}
We call symbol of $\mathcal N_\gamma$ at the point $z\in S^2$ the operator
\be\label{whynot}
\sigma[\mathcal N_{
%\mathcal N_
\gamma}](z)=\sum\limits_{k=-(N-1)}^{N-1}\widetilde\sigma_{{k;\gamma_k}}(z)
\ee
where $\widetilde\sigma_{{k;\gamma_k}}$ is given by \eqref{whyes}.
\end{definition}

\vskip 3cm

\vskip 1cm
Let us finish this section by giving a more global ``quantization" type definition of the symbol. This end of Section \ref{symbol} is not necessary for the understanding of the rest of the paper.

Note that
\begin{eqnarray}
\widetilde\sigma
_{k;\gamma_{k}}&=&
\left(
\mu_N\left(\frac{\tau(z)}\hbar-i\partial_x\right)
\frac
{e^{ix}}{\bar z}\right)^k
\gamma_k
\left({\tau(z)}-i\hbar\partial_x\right)\nonumber\\
&=&
\left(
\mu_N\left(\frac{\tau(z)}\hbar-i\partial_x\right)
\frac
{e^{i(x+\theta(z))}}{|z|}\right)^k
\gamma_k
\left({\tau(z)}-i\hbar\partial_x\right)\nonumber.\\
&=&
\frac{e^{ik\theta(z)}}{|z|^k}
\left(
\mu_N\left(\frac{\tau(z)}\hbar-i\partial_x\right)
{e^{ix}}\right)^k
\gamma_k
\left({\tau(z)}-i\hbar\partial_x\right)\nonumber.\\
&=&
\frac{e^{ik\theta(z)}}{|z|^k}
\left(
%\mu\left(\frac{\tau(z)}\hbar-i\partial_x\right)
\sqrt{\frac{\tau(z)-i\hbar\partial_x}{1-(\tau(z))-i\hbar\partial_x)}}
{e^{ix}}\right)^k
\gamma_k
\left({\tau(z)}-i\hbar\partial_x\right)\nonumber.\\
&=&
%\frac{e^{ik\theta(z)}}{|z|^k}
\left(
\frac{Z(z)}{|z|}\right)^k
\gamma_k
\left({\tau(z)}-i\hbar\partial_x\right)\nonumber.
\end{eqnarray}
Here the operator $Z(z)$ is the canonical 
(anti) 
pseudodifferential quantization of the canonical function $\mathcal Z(x,\tau):=\sqrt{\frac{\tau}{1-\tau}}e^{ix}$, ``shifted by $(\tau(z),\theta(z)$" where $z=\sqrt{\frac{\tau(z0}{1-\tau(z)}}e^{i\theta(z)}$, that is $\mathcal Z_z(\tau,x)=f(\tau+\tau(z),x+\theta(z))$. 

More precisely, the (anti) 
pseudodifferential quantization of of a function $g$ is the 
pseudodifferential quantization of $G$ where one put all the differentail part on the left (rather than on the right for the standard pseudodifferential calculus introduced at the beginning of Section \ref{intro}.

Namely, for any function $g(e^{i\theta},\tau)$ on the sphere, we define $\mbox{Op}^{APD}[g]$ and $\mbox{Op}^{APD}_z[g]$ by their integral kernels
\begin{eqnarray}
\mbox{Op}^{APD}[g](x,y)&=&
\int g(y,\tau)e^{i\tau(x-y)/\hbar}d\tau/(2\pi\hbar)\nonumber\\
\mbox{Op}^{APD}_{z}[g](x,y)&=&\int g({y+\theta(z)},\tau+\tau(z))e^{i\tau(x-y)/\hbar}d\tau/(2\pi\hbar),\nonumber
\end{eqnarray}
one has
\be\label{zzz}
Z(z)=\mbox{Op}^{APD}[\mathcal Z_z]=\mbox{Op}^{APD}_{z}[\mathcal Z].
\ee
and
\[
\widetilde\sigma_{k;\gamma_k}({z})=
\left(\mbox{Op}^{APD}_{z}
\left(
\frac{\mathcal Z_z}{|z|}
\right)
\right)^k
\mbox{Op}^{APD}_{z}(\gamma_k).
\]
%\subsection{The symbol}\label{symbol}

\begin{definition}
For any trigonometric polynomial on the sphere $\ns=\ns(e^{i\theta},\tau)=\sum\limits_ke^{ikx}\ns_k(\tau)$ we define\footnote{one can also say that $\mbox{Op}_{z}^{T}[\ns]=\ns^{PS}\left(\mbox{Op}^{APD}_{z}\left(\frac{\mathcal Z_z}{|z|}\right),
\mbox{Op}^{APD}_{z}(\tau)\right),
$
where $\ns^{PS}$ is the pseudodifferential ordering of the trigonometric polynomial $\ns$, that is the one with all the 
 $\frac{\mbox{Op}^{APD}_{z}(\mathcal Z_z)}{|z|}$ on the left.}
\[
\mbox{Op}_{z}^{}[\ns]=
\sum_k
\left(\mbox{Op}^{APD}_{z}
\left(
\frac{\mathcal Z_z}{|z|}
\right)
\right)^k
\mbox{Op}^{APD}_{z}(\ns_k).
\]
%\footnote{$\mbox{Op}_{z}^{T}[\ns]=\ns^{PS}\left(\mbox{Op}_{z_0}\left(\frac{z}{\bar {z_0}}\right),
%\mbox{Op}_{z_0}(\tau)\right),
%$
%where $\ns^{PS}$ is the pseudodifferential ordering of the trigonometric polynomial $\ns$, that is the one with all the 
% $\frac{\mbox{Op}_{z_0}(z)}{\bar{z_0}}$ on the left.}
 \end{definition}
Let us now define the ``naive" symbol of $\mathcal N$ as the function
\be\label{naive}
\ns_\mathcal N(\tau, {\theta})=\sum_{k=-(N-1)}^{N-1}e^{ik\theta}\gamma_k(\tau)=\gamma(\theta,\tau).
\ee

\begin{proposition}
\[
\sigma[\mathcal N](z)=
\mbox{Op}_{z}^{}[\ns_{\mathcal N}].
%\ns_{\mathcal N}^{APS}\left(\frac{\mbox{Op}_{z_0}(z)}{\bar{z_0}},
%\mbox{Op}_{z_0}(\tau)\right),
\]
%where $\ns_{\mathcal N}^{APS}$ is the ordering of $\ns_\mathcal N$ with all the 
% $\frac{\mbox{Op}_{z_0}(z)}{\bar{z_0}}$ on the left.
 \end{proposition}
\subsection{Symbolic calculus}\label{calculus}\ 

As a direct corollary of (the second part of) Lemma \ref{coc} we get the following result.
\begin{theorem}\label{sigmann}
\[
\sigma[{\mathcal N'\mathcal N}]=\sigma[{\mathcal N'}]\sigma[{\mathcal N}].
%+O(\hbar).
\]
\end{theorem}

\vskip 1cm
Let us define 
\be\label{cnumber}
\cof(z):=\sum_{k=-(N-1)}^{N-1}\sqrt{\frac{C^N_{\frac{\tau(z)}\hbar-(-i\partial_x)-k}}{C^N_{\frac{\tau(z)}\hbar-(-i\partial_x)}}}e^{ik\theta(z)}
\ee
and the convolution
\be\label{convol}
\sigma_{\mathcal N}\star\cof(z)
:=\int_{S^1}
\sigma_{\mathcal N}(ze^{-i\theta})
\cof
(
|z|e^{i\theta}
)
d\theta.
\ee

\begin{proposition}\label{symbbb}
\[
\mathcal N\psi^a_z
=\psi_z^{\sigma(\mathcal N)\star\cof(z)a}.
\]
\end{proposition}
\begin{proof}
By decomposition on $k$-diagonal parts of $\mathcal N$, Proposition \ref{symbbb} is a direct consequence of Proposition \ref{thefactgamma} and the fact that
$$
\Sigma_{k,\gamma_k}(z)=\widetilde\sigma_{{k,\gamma_k}}\star\cof(z).
$$
\end{proof}
\subsection{a-T\"oplitz quantization}\label{at}
\begin{definition}\label{atop}
To a (trigonometric) family $z\mapsto \Sigma(z)$ of (bounded) operators on $L^2(\R)$ we associate the operator $\mbox{Op}^T_a(\Sigma)$ on $\mathcal H_N^a$ defined by
\[
\mbox{Op}^T_a(\Sigma)=\int_{S^2}|\psi_z^{\Sigma\star\cof(z)a}\rangle_a\langle\psi^a_z|d\mu_N(z)
\]
\end{definition}
The following result is one of the main of this paper: it express that any trigonometric matrix, as defined by \eqref{defngamma}, is $a$-T\"oplitz operator, and that its $a$-T\"oplitz symbol is \textit{excatly} the symbol, as defined by \eqref{defsym}.
\begin{theorem}\label{onemain}
Let $\gamma,\gamma'$ and $\mathcal N_\gamma, \mathcal N_{\gamma'}$ as in Definition \ref{defsym}. Then
\begin{eqnarray}
\mathcal N_\gamma&=&\mbox{Op}^T_a(\sigma[{\mathcal N_\gamma}])
%\]
%\[
%\mbox{Op}^T_a(\sigma_{\mathcal N_{\gamma_1}})
%\mbox{Op}^T_a(\sigma_{\mathcal N_{\gamma_2}})
%=
\nonumber\\
\mathcal N_{\gamma'}&=&\mbox{Op}^T_a(\sigma[{\mathcal N_{\gamma'}}])
%\]
%\[
%\mbox{Op}^T_a(\sigma_{\mathcal N_{\gamma_1}})
%\mbox{Op}^T_a(\sigma_{\mathcal N_{\gamma_2}})
%=
\nonumber\\
\mathcal N_{\gamma}\mathcal N_{\gamma'}
&=&
\mbox{Op}^T_a(\sigma[{\mathcal N_{\gamma}}]\sigma[{\mathcal N_{\gamma'}}])
%+O(\hbar)
\nonumber
%\\
%\mathcal N_{\gamma_1}\dots\mathcal N_{\gamma_M}
%&=&
%\mbox{Op}^T_a(\sigma_{\mathcal N_{\gamma_1}}\dots\sigma_{\mathcal N_{\gamma_M}})
%%+O(\hbar)
%\nonumber
\end{eqnarray}
\end{theorem}
\begin{proof}
Theorem \ref{onemain} is  verbatim a straightforward consequence of Theorem \ref{sigmann}.
\end{proof}
\begin{remark}
Although we don't want to prove it here in order not to introduce too much semiclassical technicalities, let us mention that, in the case where the symbol of an $a$-T\"oplitz operator is just a regular potential (multiplication operator by a function of $x$), then one can show that the $a$-T\"oplitz operator is actually a standard T\"oplitz operator. Conversely, a standard T\:oplitz operator is an $a$-T\"oplitz operator with a symbol which is a potential.
\end{remark}
\subsection{Classical limit and underlying ``phase-space"}\label{classical}

We can rewrite the general structure of the symbol of an $a$-T\"oplitz operator $T$
has the form (near the south pole where $\tau\sim |z|\sim 0$)
\[
\sigma(z)=S(1-i\frac\hbar{\tau(z)}\partial_x,x+\theta,\tau(z)-i\hbar\partial_x,\hbar)
\]
where the function $S$ is $2\pi$ periodic in the second variable and the quantization present in the two first variables is the one of antipseudodifferentail calculus.

The function $S$ satisfies
\[
S(1+\xi,x+\theta,\tau(z)+\xi,\hbar)\to
S(1,x+\theta,\tau(z),\hbar)=\gamma(\tau(z),\theta+x)\ \mbox{ as }\xi\to 0,
\]
where $\gamma(\tau,\theta,\hbar)$ is the so-called naive symbol of $T$.

As $\hbar\to 0$, $z\neq 0$, 
\[
\sigma(z)\to
\gamma(\tau,\theta+x)
%S(1-i\frac\hbar{\tau(z)}\partial_x,x+\theta,\tau(z)-i\hbar\partial_x,0)
\]
but the limit $\hbar,z\to 0$ is multivalued. Indeed as 
\[
\left\{\begin{array}{rcl}
\hbar&\to& 0\\
z&\to& 0\\
\frac\hbar{\tau(z)}&=&\hbar_0
\end{array}
\right.
\]
we have
\[
\sigma(z)\to
S(1-i\hbar_0\partial_x,e^{i(x+\theta)},0,0).
\]
And the ``classical" noncommutative multiplication for the function $S$ is given by:
\begin{eqnarray}
S\# S'(1-\hbar_0\xi,\theta+x,\tau,0)
&=&
S(1-\hbar_0\xi,\theta+x+i\partial_{\xi'},\tau,0)
S'(1-\hbar_0\xi',\theta+x,\tau,0)|_{\xi'=\xi}\nonumber\\
&:=&
S(1-\hbar_0\xi,\theta+x+i\overset{\rightarrow}{\partial_\xi},\tau,0)
S'(1-\hbar_0\xi,\theta+x,\tau,0)\nonumber
%\overset{\rightarrow}{\partial}
\end{eqnarray}
This  define  the classical phase-space, as a noncommutative algebra of functions i.e.a noncommutative blow up of the singularity.

%Fiber bundle?
%
%In full generality (near the south pole), the limit $\hbar\to 0$ of the symbol will then be represented by a function 
%$S(\tau, 1-\hbar_0\xi,e^{i(x+\theta)}$ with multiplication
%$$S_1\#S_2(\tau,1-\hbar_0\xi,e^{i(x+\theta)})
%:=
%S_1(\tau,1-\hbar_0\xi,e^{i(x+\overset{\rightarrow}{\partial_\xi}+\theta)})
%S_2(\tau,1-\hbar_0\xi,e^{i(x+\theta)}).$$
%
%This is a function on a fiber bundle over the sphere.
%
%The $a$-T\"oplitz quantization will be then obtained by taking for $a$-T\"oplitz symbol the operator quantization in $(\xi,x)$ of the ``section" $\hbar_0=\frac\hbar\tau,\ \theta=\sqrt{\frac z{\bar z}}$:.
%
%\textit{The formal neighbourhood $\tau\sim 0$ is encapsulated in the fiber at $\tau=0$.}

%\subsection{Link with T\"oplitz quantization}\label{linktop}
%\newpage
\section{Application to TQFT}\label{tqft}
In this section we apply the results of the preceding one and show that any curve operator in TQFT of the case of the once punctured torus or the $4$-times punctured sphere. We first introduce in a very fast way curve operators. For more details, the reader can consult  \cite{mp} which is precisely referred in the next sections, and \cite{and1,and2,and3,bhmv,bp, fww, goldman, hitchin, mv, rt,tw,turaev, witten}.

%blabla

\subsection{The curve operators in the case of the once punctured torus or the $4$-times punctured sphere}\label{curve}
To any   closed oriented surface $\Sigma$ with marked points $p_1,\ldots,p_n$, any integer $r>0$ and any coloring $c=(c_1,\ldots,c_n), c_i\in\{1,\dots,r-1\}$ of the marked points, TQFT provides, by the construction of \cite{bhmv}, a finite dimensional hermitian vector space $V_r(\Sigma,c)$ together with a basis $\{\phi_n,n=1,\dots,\dim{(V_r(\Sigma,c))}\}$ of this space (see Sections 2.1 and 2.5 in \cite{mp}).

 On the other (classical) side, to each $t\in(\pi\mathbb Q)^n$ we can associate the  moduli space:
$$\mathcal{M}(\Sigma,t)=\{\rho:\pi_1(\Sigma\setminus\{p_1,\ldots,p_n\})\to\su\text{ s.t. }\forall i,\ tr\rho(\gamma_i)=2\cos(t_i)\}/\sim$$ 
where one has $\rho\sim \rho'$ if there is $g\in$ SU$_2$ such that $\rho'=g\rho g^{-1}$ and  $\gamma_i$ is any curve going around $p_i$. 

When $\Sigma$ is either a once punctured torus or a 4-times punctured sphere, $\mathcal{M}(\Sigma,t)$ is symplectomorphic to the standard sphere $S^2=\C P^1$.

To nay curve $\gamma$ on the surface $\Sigma$ (that is, avoiding the marked points) we can associate two objects: a quantum one, the curve operator $T_\gamma$ acting on $V_r(\Sigma,c)$, and a classical  one,  the function $f_\gamma$ on the symplectic manifold $\mathcal{M}(\Sigma,t)$.
\begin{itemize}
\item $T_\gamma$ is obtained by a combinatorial topological construction recalled in Sections 2.3 and 2.4 in \cite{mp}. By the identification of the finite dimensional space $V_r(\Sigma,c)$ with the Hilbert space of the quantization of the sphere $\mathcal H_N$ defined in Section \ref{canonical} with $N:=\dim{(V_r(\Sigma,c))}$, through $\{\phi_n,n=1,\dots,\dim{(V_r(\Sigma,c))}\}\leftrightarrow \{\psi^N_n,n=1,\dots,N\}$, $T_\gamma$ can be seen as a matrix on $\mathcal H_N$. One of the main results of  \cite{mp} was to prove than this matrix is a trigonometric one in the sense of Section \ref{trigomat}.
\item $f_\gamma:\ \mathcal{M}(\Sigma,t)\to [0,\pi]$ is defined by 
\be\label{fgamma}
\rho\mapsto f_\gamma(\rho):= -tr \rho(\gamma).
\ee
\end{itemize}
The asymptotism considered in \cite{mp} consists in letting $r\to\infty$ and  considering a sequence of colorings $c_r$ such that $\pi\frac{c_r}r$ converges to $t$ and the dimension of $V_r(\Sigma,c_r):=N$, grows linearly with $r$. One sees immediately that, by the identification  $V_r(\Sigma,c)\leftrightarrow \mathcal H_N$, this correspondents to the semiclassical asymptotism $N\to\infty$.  

The main result of \cite{mp} states that, {\bf for generic values of $t$}, 
$$
T_\gamma\mbox{ is a (standard) T\"oplitz operator of leading symbol }f_\gamma.
$$
the generic values of $t$ are the one for which $f_\gamma$ considered as a function on $S^2$ by the symplectic isomorphism mentioned earlier, belongs to $C^\infty(S^2)$.
\subsection{Main result}\label{mrtqft}
It is easy to see that, for the remaining non generic values of $t$, $T_\gamma$ is not a standard T\"oplitz operator. Nevertheless, it happens that it is an $a$-T\"oplitz one.

\begin{theorem}\label{main}
%[see Theorem \ref{curvetoplitz} in Section \ref{asymptoregime}]
%Suppose that $\boM(\Sigma,t)$ is smooth. Then there is a canonical diffeomorphism $\boM(\Sigma,t)\simeq \C P^1$ such that for any curve $\gamma$ on $\Sigma$, 
Let again $\Sigma$ be either the once punctured torus or the $4$-times punctured sphere. 
For \textbf{all values of $t$}, the sequence of matrices $(T_r^{\gamma})$ are the matrices in the basis $\{\psi^N_n\}|_{n=0,\dots,N-1}$ of a family of $a$-T\"oplitz operators on $\mathcal H_N^a$ with 
%(principal) 
symbol $\sigma^T_{T_r^\gamma}$ satisfying, away of the two poles,
%^\gamma=\sigma_0^\gamma+\frac 1 N\sigma_1^\gamma+O(N^{-2})$ with

\be\label{princ}
\sigma^T_{T_r^\gamma}(z)=f_\gamma(e^{-ix}z)
%\mbox{Op}_z^{APS}(-\ tr\rho(\gamma))
%\sigma_{Qf_\gamma}
+O(\sqrt\hbar)
\ee
where $f_\gamma$ is the trace function defined by \eqref{fgamma}.
% $\sigma_f$ is given by \eqref{whynot} and the operator $Q$ is defined by \eqref{propq}.

%\begin{equation}\label{sub}
%\sigma_1^\gamma=\frac{1}{2}\Delta_S\sigma_0^\gamma
%\end{equation}
%% \be\label{sub}
%%\left[\frac\tau{2(1-\tau)}\partial^2_{\tau^2}+\tau^2\partial_\tau-\frac {\partial^2_{\theta^2}}{8\tau(1-\tau)}
%%+\left(\frac {i\partial_{\theta}\partial_\tau}2-\frac{i(2\tau^2+2\tau-1)\partial_\theta}{4\tau(1-\tau)}\right)
%%\mathcal {H}_\theta\right]\sigma_0^\gamma
%%\ee
%where $\Delta_S$ is the Laplacian on the sphere which is equal to $(1+|z|^2)^2\partial_z\partial_{\ba z}$ in the canonical holomorphic coordinate $z$.
%%where $\mathcal {H}_\theta$ is the Hilbert transform acting on the  variable $\theta$.

%Namely,
%%let
%%$
%%f_\gamma(\tau, \theta)=-\ tr\rho(\gamma).
%%$
%%Then
%\[
%T^\gamma_r=\int|\psi_z^{\sigma_{Qf_\gamma}\star C a}\rangle_a\langle\psi_z^a|d\mu_N(z)+O(\sqrt\hbar).
%\]
%%with
%%\[
%%\Gamma_z=F_\gamma(\Sigma_z^++\Sigma_z^-,e^{i\hbar J}\Sigma_z^++e^{-i\hbar J}\Sigma_z^-,\hbar J).
%%\]
%Moreover, the leading symbol of $T_r^\gamma$ is determined by the trace function $f_\gamma$.
 \end{theorem}
 
% One of the meanings of Theorem \ref{main} is the fact that the classical curve function $f_\gamma$ determines, but is not equal since the T\"oplitz-paradigm has changed, the leading symbol of the curve operator.

\subsection{Proof of Theorem \ref{main}}\label{proffmain}
We give the proof in the case where $\Sigma$ is the once punctured torus. 
%\subsubsection{Decomsposition result proven in \cite{mp}}\label{mpmpmp}\ 

In \cite{mp} we proved that nay curve on $\Sigma$ is generated by the curves $\gamma,\delta,\zeta$ described in Section 3 of \cite{mp}. Therefore any curve operator belongs to the algebra generated by  three  matrices $T_r^\gamma,T_r^\delta,T_r^\zeta$ explicitely given by Proposition 3.1 and the end of Section 3.2 in \cite{mp}.

%\vskip 3cm

The explicit expressions of the matrix elements of $T_r^\gamma,T_r^\delta,T_r^\zeta$ in \cite{mp}, recalled in Section \ref{exemples},  shows clearly that these three triangular matrices are of the form \eqref{defngamma}.
% with some $\gamma_k\in C^\infty([0,1]),\ |k|\leq 1$. 
 Therefore we know from Theorem \ref{onemain} that they define three $a$-T\"oplitz operators, and therefore, again by Theorem \ref{onemain}, any curve operator  is an \at operator\footnote{Note that in \cite{mp} we proven this type of result by another method since we wanted also to define a symbol inside the interior of $\Sigma$ in the singular cases. Since we proved that in any coloring the three matrices are \at operators, we don't need such a direct result here.}.

The first assertion of Theorem \ref{main} is proven. 

Moreover, we proved in \cite{mp} that the naive symbol of $T\gamma$ as defined by \eqref{naive},  
is equal. out of the poles and modulo $\hbar$, to the trace function. Using now \eqref{ddefmu}, \eqref{munu} and \eqref{whyes} we find easily that as $\hbar\sim 0$,
\[
\gamma_k(\tau(z)-i\hbar\partial_x)\mu_k\left(\frac{\tau(z)}\hbar-i\partial_x\right)\frac{e^{ikx}}{\bar z^k}
\sim
\gamma_k(\tau(z))\left(\sqrt{\frac{\tau(z)}{1-\tau(z)}}\right)^k\frac{e^{ikx}}{\bar z^k}
%=\gamma_k(\tau(z))\frac{|z|}{\bar z}e^{ikx}
=\gamma_k(\tau(z))e^{-ik(\theta(z)-x},
\]

which gives \eqref{princ} after summation on $k$. 
\subsection{Examples}\label{exemples}

The symbol of a general curve operator is given by \eqref{whyes} out of is matrix elements. But we were not able in this paper to rely it, near the poles, to the trace functions in general. But in the cases of the three operators $T_r^\gamma,T_r^\delta,T_r^\zeta$ we can do it.

We find, by Proposition 3.1 and the relabelling of indices in the item 1. in Section 4.2 in \cite{mp}, that

\begin{eqnarray}
T_r^\gamma\psi_n&=&
-2\cos{(\tfrac\pi N(n+\tfrac{a+1}2))}\psi_n\nonumber\\
T_r^\delta\psi_n&=&u_{n+1}\psi_{n+1}+u_n\psi_{n-1}\nonumber\\
T_r^\zeta\psi_n&=&u_{n+1}e^{i(\tfrac\pi N(n-a))}\psi_{n+1}+u_ne^{-i\tfrac\pi N(n+a)}\psi_{n-1}\nonumber,
\end{eqnarray}
where 
%$0\leq\alpha<\pi$ (the limiting coloring) 
$a\in\mathbb{N}, a\mbox{ odd}$, is the color assigned to the marked point 
and 
\[
u_n=-\left(\frac{\sin{(\tfrac\pi N(n+a)}\sin{(\tfrac\pi N n)}}{\sin{\tfrac\pi N(n+\tfrac{a+1}2)}\sin{\tfrac\pi N(n+\tfrac{a-1}2)}}\right)^{\frac12}.
\]
Defining as before $\hbar=\frac\pi N$, and $\alpha=\hbar a$, we obtain that the naive symbols of $T_r^\gamma,T_r^\delta,T_r^\zeta$, as defined in \eqref{naive},  are ($\tau=n\hbar$)
\begin{eqnarray}
\sigma_{T_r^\gamma}(\tau,\theta,\hbar)&=&
-2\cos{(\tau+\tfrac{\alpha+\hbar}2)}
\nonumber\\
\sigma_{T_r^\delta}(\tau,\theta,\hbar)&=&
-\left(\frac{\sin(\tau+\alpha+\hbar)\sin(\tau+\hbar)}{\sin(\tau+\tfrac{\alpha+\hbar}2)\sin(\tau+\tfrac{\alpha+3\hbar}2)}\right)^{1/2}e^{i\theta}-\left(\frac{\sin(\tau+\alpha)\sin(\tau)}{\sin(\tau+\tfrac{\alpha+\hbar}2)\sin(\tau+\tfrac{\alpha-\hbar}2)}\right)^{1/2}e^{-i\theta}
\nonumber\\
\sigma_{T_r^\zeta}(\tau,\theta,\hbar)&=&e^{i\tfrac\hbar 2}\sigma_{T_r^\delta}(\tau,\theta+\tau,\hbar)\nonumber
\end{eqnarray}

and, as $\hbar\to 0$,
\begin{eqnarray}
\sigma_{T_r^\gamma}(\tau,\theta,0)&=&
-2\cos{(\tau+\tfrac{\alpha}2)}
\nonumber\\
\sigma_{T_r^\delta}(\tau,\theta,0)&=&
-\left(\frac{\sin(\tau+\alpha)\sin(\tau)}{\sin(\tau+\tfrac{\alpha}2)\sin(\tau+\tfrac{\alpha}2)}\right)^{1/2}e^{i\theta}-\left(\frac{\sin(\tau+\alpha)\sin(\tau)}{\sin(\tau+\tfrac{\alpha}2)\sin(\tau+\tfrac{\alpha}2)}\right)^{1/2}e^{-i\theta}
\nonumber\\
\sigma_{T_r^\zeta}(\tau,\theta,0)&=&\sigma_{T_r^\delta}(\tau,\theta+\tau,0)\nonumber.
\end{eqnarray}
We see that, for all values of $\alpha$, $\sigma_{T_r^\gamma}(\tau,\theta,0)\in C^\infty(S_2)$. For $\alpha>0$, 
$$
\sigma_{T_r^\delta}(\tau,\theta,0)\sim -\sqrt\tau\cos{\theta},\mbox{ as }\tau\sim 0,
$$
so that. for $\alpha>0$, $\sigma_{T_r^\delta}(\tau,\theta,0),\ \sigma_{T_r^\zeta}(\tau,\theta,0)\in C^\infty(S_2)$.

But when $\alpha=0$ then 
\begin{eqnarray}
\sigma_{T_r^\delta}(\tau,\theta,0)&=&-2\cos{\theta}\nonumber\\
\sigma_{T_r^\zeta}(\tau,\theta,0)&=&-2\cos{(\theta+\tau)}\nonumber
\end{eqnarray}
so that $\sigma_{T_r^\delta}(\tau,\theta,0),
\sigma_{T_r^\zeta}(\tau,\theta,0)$ are singular on the sphere. Note that, for $\alpha=0$, the corresponding moduli space $\mathcal M$ is also singular (see Remark 4.13 in \cite{mp}).

In fact $T_r^\delta=-M_1$ where $M_1$ is precisely the toy matrix  defined in Section \ref{toy} above. The same way, $T_r^\zeta=-M_{1,\kappa}$ as defined in Section \ref{trigomat} with $\kappa(\tau)= e^{i\tau}$. Therefore its symbol, together with the one of $T_r^\zeta$, is given out of the trace functions of $\delta$ and $\zeta$ by Definition \ref{defsym}. 
In other words,
$$
T_r^\eta=\mbox{Op}^T_a(\sigma[\mathcal N_{f_\eta}])+O(\hbar),\ \eta=\gamma,\delta,\zeta.
$$
 where $f_\eta=-tr{\rho(\eta)}$ is the trace function defined by \eqref{fgamma}.
 
Straightforward but tedious computations show that the symbols of all the curve operators on the 4th-punctured sphere are also given out of the corresponding trace functions by Definition \ref{defsym}. This suggest the conjecture in the following section.
\newtheorem{conjecture}{Conjecture}
\section{A conjecture}
\textbf{Conjecture}: {\it
Any curve operator  $T^\gamma_r$ is an $a$-T\"oplitz operator whose symbol satisfies
$$
\sigma^T_{ T^\gamma_r }=\sigma[\mathcal N_{-tr\rho(\gamma)}]+O(\hbar).
$$
in the sense of Definition \ref{defsym}.

\vfill
\end{document}